\documentclass[amsmath,amssymb,superscriptaddress]{revtex4}
\usepackage[english]{babel}
\usepackage{amssymb,latexsym}
\usepackage{amssymb}
\usepackage{amsmath}
\usepackage[latin1]{inputenc}
\usepackage{pst-3d}
\usepackage{pb-diagram}
\usepackage{graphicx}

\definecolor{gris}{cmyk}{0,0,0,0.15}
\usepackage[english]{babel}
\usepackage{graphics}
\usepackage{amsmath}
\usepackage{epsfig}
\usepackage{enumerate}
\usepackage{amsthm}
\usepackage{hyperref}
\usepackage{pb-diagram}

\newtheorem{thm}{Theorem}[section]

\newtheorem{lemma}{Lemma}[section]

\usepackage{fancyhdr}

\begin{document}
\begin{flushright} CUQM-161~~~~~~~\end{flushright}
% -------------------------------------------------------------------------------------------------------
\title{Refining the general comparison theorem for the Klein--Gordon equation}
% -------------------------------------------------------------------------------------------------------
\author{Richard L. Hall}
\email{richard.hall@concordia.ca}
\affiliation{Department of Mathematics and Statistics, Concordia University,
1455 de Maisonneuve Boulevard West, Montr\'eal,
Qu\'ebec, Canada H3G 1M8}
\author{Hassan Harb}
\email{hassan.harb@concordia.ca}
\affiliation{Department of Mathematics and Statistics, Concordia University,
1455 de Maisonneuve Boulevard West, Montr\'eal,
Qu\'ebec, Canada H3G 1M8}

\begin{abstract}

By recasting the Klein--Gordon equation as an eigen-equation in the coupling parameter $v > 0,$ the basic Klein--Gordon comparison theorem may  be  written  $f_1\leq f_2\implies G_1(E)\leq G_2(E)$, where $f_1$ and $f_2$, are the monotone non-decreasing shapes of two central potentials $V_1(r) = v_1\,f_1(r)$ and $V_2(r)  = v_2\, f_2(r)$ on $[0,\infty)$. Meanwhile $v_1 = G_1(E)$ and $v_2 = G_2(E)$ are the corresponding coupling parameters that are functions of the energy $E\in(-m,\,m)$. We weaken the sufficient condition for the ground-state spectral ordering by proving (for example in $d=1$ dimension) that if 
$\int_0^x\big[f_2(t) - f_1(t)\big]\varphi_i(t)dt\geq 0$, the couplings remain ordered $v_1 \leq v_2$ where $i = 1\, {\rm or}\,  2, $ and  $\{\varphi_1, \varphi_2\}$  are the ground-states corresponding respectively to the couplings $\{v_1,\, v_2\}$ for a given $E \in (-m,\, m).$. This result is extended to spherically symmetric radial potentials in $ d > 1 $ dimensions.
\end{abstract}

\maketitle
\vskip 0.2 in
% --------------------------------------------------------------------------%
\section{Introduction}
% --------------------------------------------------------------------------%
The elementary comparison theorem of non-relativistic quantum mechanics states that if two potentials are ordered, then the respective discrete eigen energies are correspondingly ordered. When the Hamiltonian $H = -\Delta + V$ is bounded from below, this theorem may be established in the Schr\"odinger case as a direct result of the min-max variational principle \cite{ReedSimon}. This is important because it allows immediately for spectral approximations in terms of known exact solutions of comparable problems.     A variational principle in such a simple form in the relativistic case would appear at first sight  to be unattainable since the energy operators are not bounded from below \cite{Fr, Gold, Grant}. Exact analytical solutions of the Klein--Gordon equation for square-well, exponential, Woods-Saxon \cite{RDWoods}, and Cusp potentials are presented in \cite{Gnr1, MBL, SSW, VMC}, and  cut-off Coulomb potentials were solved analytically by Barton \cite{Barton} and  subsequently generalized and confirmed numerically by Hall \cite{Hallcutoff}. It is clear from these examples that the graphs of $E(v)$ are certainly not monotone. However, by the use of monotonic analysis several relativistic comparison theorems have been  established \cite{SpecDirac, Hallrel, HallDirac, Aliyu, GChen, GChen1}, but the proofs of such results for the Klein--Gordon equation were all restricted to positive energies. In a recent paper \cite{HH}  we were able to generalize these comparison theorems by studying the eigen--problem in the coupling parameter $v$ for a given $E \in (-m,\, m)$.  The spectral  relations expressed in form $v = G(E)$ are single-valued {\it functions}, which facilitate the design of the unrestricted Klein--Gordon comparison theorems. 
\medskip

In the present paper we consider the Klein--Gordon relativistic equation with an attractive central potential $V(x) = v\,  f(x)$ for a single particle. We proved in \cite{HP-16} that if two potential shapes are ordered, $f_1(x)\leq f_2(x)$, then the corresponding spectral curves are similarly ordered. In the present paper we tighten the condition for this theorem by proving that the ordering of the coupling parameters is still preserved for the ground state, when the graphs of the potential shapes cross over, provided $\int_0^x\big[f_2(t) - f_1(t)\big]dt\geq 0$. Moreover, if one of the wave functions $\varphi_1$ or $\varphi_2$ is known, we prove that $\int_0^x\big[f_2(t) - f_1(t)\big]\varphi_i(t)dt\geq 0\implies G_1(E)\leq G_2(E),$ $i = 1,2$. This is a stronger version of the  theorem because, since the ground state is non-increasing on $[0,\infty)$, this allows the potential shapes to cross over `even more' with the ordering of couplings still preserved. The idea of refining the comparison theorems of non-relativistic and relativistic quantum mechanics has been first presented in Ref.\cite{HRQ}, and was applied for the Dirac equation \cite{Dirac1, Dirac2}, and the Klein--Gordon equation \cite{HP-16}, but the latter case was restricted to non-negative energies. The present work removes the restriction to positive energies and the results are obtained for attractive central potentials in all dimensions $d\geq 1$.

% --------------------------------------------------------------------------%
\section{Refined Theorems.}
% --------------------------------------------------------------------------%

% --------------------------------------------------------------------------%
\subsection{One dimensional case}
% --------------------------------------------------------------------------%
The Klein--Gordon equation in one dimension is given by \cite{PS, WG}:
\begin{eqnarray}\label{KG}
\varphi^{\prime\prime}(x) =\big[m^2-\big(E-V(x)\big)^2\big]  \varphi(x),\quad x\in\mathbb{R}.
\end{eqnarray}
 where $\varphi^{\prime\prime}$ denotes the second order derivative of $\varphi$ with respect to $x$, natural units $\hbar=c=1$ are used, and $E$ is the energy of a spinless particle of mass $m$.
 We suppose that the potential function $V$ is expressed as $V(x)=v f(x)$ with $v>0$ and $f$ satisfies the following conditions:
\begin{enumerate}
\item  $V(x) = v f(x), x\in\mathbb{R}$, where $v>0$ is the coupling parameter and $f(x)$ is the potential shape;
\item  $f$ is even $f(x) = f(-x)$;
\item  $f$ is not identically zero, and is non-positive, that is $f(x)\leq 0$;
\item  $f$ is attractive, that is $f$ is monotone non-decreasing over $[0,\infty)$;
\item  $f$ vanishes at infinity, i.e $\displaystyle\lim_{x\to \pm \infty} f(x)= 0$.
\end{enumerate}
We also assume that $V(x)=v f(x)$ is in this class $\mathcal{P}$ of potentials, for which the Klein--Gordon equation (\ref{KG}) has at least one discrete eigenvalue $E$, and that equation (\ref{KG}) is the eigen-equation for the eigenstates.
Because of condition $5$, equation (\ref{KG})  has the asymptotic form
\begin{eqnarray*}
\varphi^{\prime\prime}=(m^2-E^2)\varphi,
\end{eqnarray*}
at infinity, with solutions $\varphi(x)=C_1 e^{\sqrt{k} |x|}+C_2 e^{-\sqrt{k} |x|}$, where $C_1$ and $C_2$ are constants of integration, and $k=m^2-E^2$.
The radial wave function of $\varphi$ vanishes at infinity; thus, $C_1=0$. Since $\varphi \in L^2(\mathbb{R})$, then $k>0$ which means that
\begin{eqnarray}\label{cond}
 |E|<m,
\end{eqnarray}
which is also a result discussed by Greiner \cite{WGG}. Suppose that $\varphi(x)$ is a solution of (\ref{KG}). Then by direct substitution we conclude that $\varphi(-x)$ is another solution of (\ref{KG}). Thus, by using linear combinations, we see that all the solutions of this equation may be assumed to be either even or odd.
Hence, if $\varphi$ is even then $\varphi^\prime(0) = 0$, and if $\varphi$ is odd then $\varphi(0) = 0$.
Since $\varphi\in L^2(\mathbb{R})$ then $\int_{-\infty}^{+\infty}\varphi^2 dx <\infty$. This means that the wave functions can be normalized and consequently we shall assume that $\varphi$ satisfies the normalization condition
\begin{eqnarray}\label{norm}
 ||\varphi||^2=\int_{-\infty}^{\infty} \varphi^2(x)  dx=1.
\end{eqnarray}
%\begin{definition}
%We denote by $\langle f\rangle$ and $\langle f^2\rangle$ the mean values of $f$ and $f^2$ respectively, where
%$\langle\varphi,\psi\rangle = \int_{-\infty}^{\infty}\varphi(x)\psi(x) dx$ is the inner product on $L^2(\mathbb{\R})$,  that is    $\langle f\rangle=\langle\varphi,f\varphi\rangle=\int_{-\infty}^\infty f(x)\varphi^2(x)dx$ and      $\langle f^2\rangle=\langle\varphi,f^2\varphi\rangle=\int_{-\infty}^\infty f^2(x)\varphi^2(x) dx$.
%\end{definition}
\begin{lemma}
\begin{eqnarray}\label{mineq}
E\displaystyle\int_0^{\infty} f(x)\varphi^2(x) dx \,  <  \, v \displaystyle\int_0^{\infty} f^2\varphi^2(x)dx, \quad \forall \, |E| \, <m.
\end{eqnarray}
\end{lemma}
\begin{proof} Expanding equation (\ref{KG}) we get:
\begin{eqnarray*}
\varphi^{\prime\prime}(x)=(m^2-E^2)\varphi(x) + 2Evf(x)\varphi(x) - v^2 f^2(x)\varphi(x).
\end{eqnarray*}
Multiplying both sides by $\varphi$ and integrating over $[0,+\infty)$ we obtain:
\begin{eqnarray*}
\displaystyle\int_0^{\infty} \varphi^{\prime\prime}(x)\varphi(x) dx=\dfrac{1}{2}(m^2 - E^2)+  2Ev\displaystyle\int_0^{\infty} f\varphi^2(x)dx - v^2\displaystyle\int_0^{\infty} f^2\varphi^2(x) dx.
\end{eqnarray*}
After applying integration by parts and using the fact that $(\varphi\varphi^{\prime})\bigg|_0^{\infty} = 0$ for any $\varphi$, the left-hand side of the last equation becomes $-\displaystyle\int_0^\infty \big(\varphi^{\prime}(x)\big)^2dx$.
Thus
\begin{eqnarray*}
 2E\displaystyle\int_0^{\infty} f(x)\varphi^2(x) dx -  v \displaystyle\int_0^{\infty} f^2\varphi^2(x)dx = -\int_0^{\infty} \big(\varphi^{\prime}(x)\big)^2dx+\dfrac{1}{2}(E^2-m^2) < 0.
\end{eqnarray*}
\begin{enumerate}
\item If $E\geq0$, then the result follows immediately;
\item If $E < 0$, then since
\begin{eqnarray*}
E\displaystyle\int_0^{\infty} f(x)\varphi^2(x) dx < 2E\displaystyle\int_0^{\infty} f(x)\varphi^2(x) dx,
\end{eqnarray*}
we get the desired result.
\end{enumerate}
\end{proof}
\medskip
\noindent We now define the operator $K$ as:
\begin{eqnarray}\label{op}
K=-\frac{\partial^2}{\partial x^2}+2Evf-v^2f^2.
\end{eqnarray}
If $\varphi$ is solution of the Klein--Gordon equation (\ref{KG}), then we have:
\begin{eqnarray}\label{KG2}
K\varphi = (E^2-m^2)\varphi,
\end{eqnarray}
and it follows 
\begin{eqnarray}\label{KG3}
\langle K\rangle = \langle\varphi,K\varphi\rangle=\langle\varphi,(E^2-m^2)\varphi\rangle=E^2-m^2,
\end{eqnarray}
where $\langle\varphi,\psi\rangle = \displaystyle\int_{-\infty}^{\infty} \varphi(x)\psi(x)dx$ for all $\varphi,\psi\in L^2(\mathbb{R})$.
We observe that $K$ is symmetric, that is to say: $\langle\varphi, K\psi\rangle = \langle K\varphi, \psi\rangle$.\\

\noindent We now consider the parameter $a\in [0,1]$ and the two potential shapes $f_1$ and $f_2$ with $f = f(a,x) = f_1(x) + a\big[f_2(x) - f_1(x)\big]$, where $f_1\leq f_2 \leq 0$.
Hence $f$ is non-positive, attractive, even, and vanishes at infinity. We note that $f(0,x) = f_1(x)$ when $a=0$, and  $a = 1$ when $f(1,x) = f_2(x)$,
and
\begin{eqnarray}\label{rel}
\frac{\partial f}{\partial a} = f_2(x) -f_1(x)\geq 0.
\end{eqnarray}
 Hence, $f$ is monotone non-decreasing in the parameter $a$.
Let $v$ depend on $a$ and  $E$ is be a constant, that is $v = v(a)$ and $\frac{\partial E}{\partial a} = 0$, and $-m < E < m$.
We again consider the symmetric operator $K$ in (\ref{op}), and we define $\varphi_a$ to be the partial derivative of $\varphi$ with respect to $a$.
Differentiating equation (\ref{KG3}) with respect to $a$ we get:
\begin{eqnarray}\label{diff1}
\langle\varphi_a,K\varphi\rangle + \langle\varphi,K_a\varphi\rangle + \langle\varphi,K\varphi_a\rangle = 0
\end{eqnarray}
Applying the partial derivative with respect to $a$ to equation (\ref{norm}) and using the symmetry of $K$, we obtain the new orthogonality relation
\begin{eqnarray*}
\langle\varphi_a,K\varphi\rangle = \langle\varphi, K\varphi_a\rangle = (E^2-m^2)\langle\varphi_a,\varphi\rangle = 0.
\end{eqnarray*}
We also have:
\begin{eqnarray}\label{Ka}
K_a = 2Ev_af + 2Ev(f_2 - f_1) - 2vv_af^2 - 2v^2f(f_2-f_1),
\end{eqnarray}
with $v_a$ defined as $\frac{\partial v}{\partial a}$.
Equation (\ref{diff1}) becomes:
\begin{eqnarray*}
Ev_a\displaystyle\int_0^{\infty} f\varphi^2(x) dx + Ev\int_0^{\infty}\bigg(f_2(x) - f_1(x)\bigg)\varphi^2(x)dx - vv_a\displaystyle\int_0^{\infty} f^2\varphi^2(x)dx\\
 - v^2\int_0^{\infty}f\bigg(f_2(x) - f_1(x)\bigg)\varphi^2(x) dx= 0.
\end{eqnarray*}
This leads us to the following relation:
\begin{eqnarray}\label{rel1}
v_a = \dfrac{v I}{E\displaystyle\int_0^{\infty} f(x)\varphi^2(x) dx - v \displaystyle\int_0^{\infty} f^2\varphi^2(x)dx},
\end{eqnarray}
where
\begin{eqnarray}\label{I}
I =\displaystyle\int_0^{\infty}\bigg(f_2(x) - f_1(x)\bigg) \bigg(vf(x) - E\bigg)\varphi^2(x) dx.
\end{eqnarray}
\begin{lemma}
The ground state eigenfunction $\varphi$ of the Klein--Gordon equation is a non-increasing function for $x\in[0,\infty)$, and for any energy $E$ such that  $|E| < m$.
\end{lemma}
\begin{proof}
Since $\varphi$ is an even state, then $\varphi^{\prime}(0) = 0$, and since \cite{HH} $\varphi$ is concave on $\big[0,V^{-1}(E-m)\big)$ and convex on $\big[V^{-1}(E-m),\infty\big)$, then $\varphi^{\prime}(x)\leq 0,\quad x\in[0,\infty)$.
\end{proof}
\begin{thm}
For any two potentials $f_1$, $f_2\in\mathcal{P}$ we have:
\begin{eqnarray}\label{0}
\mu(x) = \displaystyle\int_0^x[f_2(t) - f_1(t)]dt\geq 0\quad x\in[0,\infty)\implies G_1(E)\leq G_2(E),
\end{eqnarray}
for any ground state energy $E$.
\end{thm}
\begin{proof}
Applying integration by parts to (\ref{I}) we get
\begin{eqnarray*}
I = (vf(x) - E)\varphi^2(x)\mu(x)\bigg|_0^{\infty} - \displaystyle\int_0^{\infty}\big[vf^{\prime}(x)\varphi^2(x) + 2(vf(x) - E)\varphi(x)\varphi^{\prime}(x)\big]\mu(x) dx.
\end{eqnarray*}
Regarding that $\displaystyle\lim_{x\to\infty}\varphi(x) = 0$ and $\mu(0) =0$, we get
\begin{eqnarray*}
I =  - \displaystyle\int_0^{\infty}\big[vf^{\prime}(x)\varphi^2(x) + 2(vf(x) - E)\varphi(x)\varphi^{\prime}(x)\big]\mu(x) dx.
\end{eqnarray*}
Since \cite{HH}
\begin{eqnarray}\label{expE}
 E = vf(x) + \sqrt{m^2 - \dfrac{\varphi^{\prime\prime}(x)}{\varphi(x)}},
\end{eqnarray}
 $\varphi^{\prime}(x) \leq 0$, and
\begin{eqnarray*}
f^{\prime}(x) = \dfrac{\partial f}{\partial x} = (1 - a)f_1^{\prime}(x) + af_2^{\prime}(x) \geq 0,
\end{eqnarray*}
then, $I\leq 0$.\\
Therefore, following { \bf Lemma II.1.} $v_a\geq 0$, and the theorem is proven.
\end{proof}
\noindent This theorem allows us to say that graphs of $f_1$ and $f_2$ cross over in such a way that preserves the positivity of $\mu (x)$, then the corresponding coupling constants are ordered as $v_1\leq v_2$ for each $E\in (-m,\,m)$.\\

\noindent We now state a stronger version of the above theorem, which can be applied in case we know one of the ground states $\varphi_1$ or $\varphi_2$:
\begin{thm}
For any potentials $f_1$, $f_2\in\mathcal{P}$ we have:
\begin{eqnarray*}
\rho(x) = \displaystyle\int_0^x[f_2(t) - f_1(t)]\varphi_j(t)dt\geq 0\quad x\in[0,\infty)\implies G_1(E)\leq G_2(E),
\end{eqnarray*}
for $j = 1, 2$ and for any ground state energy $E$.
\end{thm}
\begin{proof}
Suppose, {\it w.l.g}, that $j = 1$. Applying the operator $\dfrac{\partial}{\partial a}$ to the expression
\begin{eqnarray}\label{K}
 K\varphi = (E^2 - m^2)\varphi
\end{eqnarray}
we get
\begin{eqnarray}\label{exp1}
K_a\varphi +K\varphi_a = (E^2 - m^2)\varphi_a,
\end{eqnarray}
where $\varphi_a = \dfrac{\partial\varphi}{\partial a}$.\\
We then multiply (\ref{exp1}) by $\varphi_1$ and apply the inner product to get
\begin{eqnarray*}
\langle\varphi_1,K_a\varphi\rangle =- \langle\varphi_1, K\varphi_a\rangle + \langle(E^2 - m^2) \varphi_1, \varphi_a\rangle,
\end{eqnarray*}
which implies that
\begin{eqnarray}\label{dff}
\langle\varphi_1,K_a\varphi\rangle =- \langle\varphi_1, K\varphi_a\rangle + \langle K\varphi_1, \varphi_a\rangle.
\end{eqnarray}
Since $K$ is symmetric, then $\langle\varphi_1,K\varphi_a\rangle = \langle K\varphi_1,\varphi_a\rangle$. Then relation (\ref{dff}) becomes
\begin{eqnarray}\label{exp2}
\langle\varphi_1,K_a\varphi\rangle = 0.
\end{eqnarray}
Using (\ref{Ka}) we obtain
\begin{eqnarray}\label{va}
v_a = \dfrac{v\left\langle\varphi_1, (f_2 - f_1)(vf - E)\varphi\right\rangle}{\left\langle\varphi_1, f(E - vf)\varphi\right\rangle}.
\end{eqnarray}
Using (\ref{expE}) we
 observe that the denominator of (\ref{va}) is negative. Applying integration by parts to the numerator changes it into
\begin{eqnarray}\label{num}
\big(vf(x) - E\big)\varphi(x)\rho(x)\bigg|_0^{\infty} - \displaystyle\int_0^{\infty}\big[vf^{\prime}(x)\varphi(x) + \big(vf(x) - E\big)\varphi^{\prime}(x)\big]\rho(x)dx.
\end{eqnarray}
 Since $\displaystyle\lim_{x\to\infty}(x) = 0$ and $\rho(0) = 0$ then (\ref{num}) becomes
\begin{eqnarray*}
- \displaystyle\int_0^{\infty}\bigg[vf^{\prime}(x)\varphi(x) + \big(vf(x) - E\big)\varphi^{\prime}(x)\bigg]\rho(x)dx\leq 0.
\end{eqnarray*}
Therefore $v_a\geq 0$ and the proof is complete.
\end{proof}

%\newpage
% --------------------------------------------------------------------------%
\subsection{$d$--dimensional cases ($d\geq 2$)}
% --------------------------------------------------------------------------%
The Klein--Gordon equation in $d$ dimensions is given by 
\begin{eqnarray}\label{KGPSI}
\Delta_d\Psi (r)=[m^2-(E- V(r))^2]\Psi (r),
\end{eqnarray}
where natural units $\hbar=c=1$ are used and $E$ is the discrete energy eigenvalue of a spinless particle of mass $m$.
We suppose here that the vector potential function $V(r)$, $r=||\bf{r}||$, is a radially-symmetric Lorentz vector potential (the time component of a space-time vector), which belongs to the class $\mathcal{P}_d$ with the following properties:
\begin{enumerate}
\item  $V(r) = v f(r), r\in\mathbb[0,\infty)$, where $v>0$ is the coupling parameter and $f(r)$ is the potential shape;
\item  $f$ is not identically zero and non-positive;
\item  $f$ is attractive, that is $f$ is monotone non-decreasing over $[0,\infty)$;
\item $f$ is not more singular than $r^{-(d - 1)}$, $r\in [0,\infty)$, that is $\displaystyle\lim_{r\to 0}r^{(d - 2)}f(r) = A,\quad -\infty < A \leq 0$;
\item  $f$ vanishes at infinity, i.e $\displaystyle\lim_{r\to \infty} f(r)= 0$.
\end{enumerate}
This is a wider potential class than $\mathcal{P}$, since it contains Coulomb and Coulomb--like potentials, such as the Yukawa and the Hulth\'en potentials. The operator $\Delta_d$ is the $d$-dimensional Laplacian. Hence, the wave function for $d > 1$ can be expressed as $\Psi (r) =R(r) Y_{l_{d-1,...,l_1}} (\theta_1, \theta_2,...,\theta_{d-1})$, where $R\in L^2 (\mathbb{R}^d)$ is a radial function and $Y_{l_{d-1,...,l_1}}$ is a normalized hyper-spherical harmonic with eigenvalues $l(l+d-1)$, $l = 0, 1, 2, ...$ \cite{rtd}
The radial part of the above Klein--Gordon equation can be written as:
\begin{eqnarray*}
\frac{1}{r^{d-1}}\frac{\partial}{\partial r}\bigg(r^{d-1}\frac{\partial}{\partial r}R(r)\bigg) = \bigg[m^2-\big(E-V(r)\big)^2+\frac{l(l+d-2)}{r^2}\bigg]R(r),
\end{eqnarray*}

where $R$ satisfies the second-order linear differential equation
\begin{eqnarray}\label{KGD}
R^{\prime\prime}(r) +\frac{d-1}{r} R^{\prime} (r)= \bigg[m^2-\big(E-V(r)\big)^2+\frac{l(l+d-2)}{r^2}\bigg]R(r).
\end{eqnarray}

Since $V$ vanishes at $\infty$, equation (\ref{KGD}) becomes 
\begin{eqnarray*}
\varphi^{\prime\prime}=(m^2-E^2)\varphi
\end{eqnarray*}
near infinity, which means that $|E| <  m$.
 The normalization condition for bound states is
\begin{eqnarray}\label{nrm}
\int_{0}^{\infty}R^2 (r) r^{d-1}dr= 1.
\end{eqnarray}
Differentiating (\ref{nrm}) with respect to $a$ we obtain the orthogonality relation $\langle R_a, R\rangle = \langle R, R_a\rangle = 0$.
%\begin{definition}
%We denote by $\langle f\rangle$ and $\langle f^2\rangle$ the mean values of $f$ and $f^2$ respectively, for $R\in L^2(\mathbb{R^d})$, where
%$\langle f\rangle=\langle\varphi,f\varphi\rangle=\int_0^{+\infty} f(r)r^{d-2}R^2(r)dr$ and $\langle f^2\rangle=\langle\varphi,f^2\varphi\rangle=\int_0^{+\infty} f^2(r)r^{d-2}R^2(r) dr$.
%\end{definition}
We also define $f(r,a) = a\,f_1(r) + (1 - a)f_2(r)$, $f_1, f_2\in\mathcal{P}_d$, and we consider the operator 
\begin{eqnarray}\label{op2}
K=-\frac{\partial^2}{\partial r^2}-\dfrac{\partial}{\partial r}+2Evf-v^2f^2.
\end{eqnarray}
By the same reasoning for the one-dimensional case we obtain the relation
\begin{eqnarray*}
v_a = \dfrac{v I}{E\langle f\rangle - v\langle f^2\rangle},
\end{eqnarray*}
where
\begin{eqnarray}\label{I2}
I =\displaystyle\int_0^{\infty}\bigg(f_2(r) - f_1(r)\bigg) \bigg(vf(r) - E\bigg) r^{(d-1)}R^2(r) dr, 
\end{eqnarray}
$\langle f \rangle = \displaystyle\int_0^{\infty}f(r) R^2(r)  r^{d-1} dr$, and $\langle f^2 \rangle = \displaystyle\int_0^{\infty}f^2(r) R^2(r)  r^{d-1} dr$.\\
\medskip
Using \cite{HH}
\begin{eqnarray}\label{ET}
E = vf(r) + \sqrt{m^2 - \frac{R^{\prime\prime}(r)}{R(r)} - \frac{R^\prime(r)}{rR(r)}},
\end{eqnarray}
we get
\begin{eqnarray}\label{rel2}
v_a = \dfrac{vI}{\displaystyle\int_0^{\infty}\bigg[\sqrt{m^2 - \dfrac{R^{\prime\prime}(r)}{R(r)} - \dfrac{R^\prime(r)}{rR(r)}}R^2(r)\bigg]r^{d-1}f(r)dr}.
\end{eqnarray}
\begin{lemma}
The ground state eigenfunction of the Klein--Gordon equation is non-increasing for $r\in[0,\infty)$ and $|E| < m$.
\end{lemma}
\begin{proof}
For $l = 0$, equation (\ref{KGD}) can be written as
\begin{eqnarray}\label{var}
R^{\prime}(r) = r^{-(d-1)}\displaystyle\int_0^rF(t)R(t)t^{d - 1}dt
\end{eqnarray}
where
\begin{eqnarray*}
F(t) = m^2 - \big(E - vf(t)\big)^2.
\end{eqnarray*}
Replacing $E$ by the expression (\ref{ET}) and using this in $F^{\prime}(t) = \dfrac{dG}{dt}$ we find
\begin{eqnarray*}
F^{\prime}(t) = 2vf^{\prime}(t)\big(E - vf(t)\big) = 2\bigg[\sqrt{m^2 - \frac{R^{\prime\prime}(t)}{R(t)} - \frac{R^\prime(t)}{tR(t)}}\bigg]vf^{\prime}(t) \geq 0.
\end{eqnarray*}
Thus we have reached the same result as in \cite{HP-16}, but extended to $|E| < m$. Hence, $R^{\prime}(r)\leq 0$ for all $r\in[0,\infty)$ and $|E| < m$.
\end{proof}
\begin{thm}
If $f_1$, $f_2\in \mathcal{P}_d$ such that $(f_2 - f_1)$ has $t^{d-1}$-weighted area, then:
\begin{eqnarray}\label{new1}
\eta(r) = \displaystyle\int_0^r[f_2(t) - f_1(t)]t^{d-1}dt\geq 0\quad r\in[0,\infty)\implies G_1(E)\leq G_2(E),
\end{eqnarray}
where $E$ is the ground state energy.
\end{thm}
\begin{proof}
Integrating (\ref{I2}) by parts we get
\begin{eqnarray*}
I = (vf(r) - E)R^2(r)\eta(r)\bigg|_0^{\infty} - \displaystyle\int_0^{\infty}\bigg[vf^{\prime}(r)R^2(r) + 2\big(vf(r) - E\big)R(r)R^{\prime}(r)\bigg]\eta(r) dr.
\end{eqnarray*}
Using $\displaystyle\lim_{x\to\infty}R(r) = 0$ and $\eta(0) = 0$ we obtain
\begin{eqnarray*}
I = - \displaystyle\int_0^{\infty}\bigg[vf^{\prime}(r)R^2(r) + 2\big(vf(r) - E\big)R(r)R^{\prime}(r)\bigg]\eta(r) dr.
\end{eqnarray*}
Hence, relation (\ref{rel2}) is non-negative and the theorem is proved.
\end{proof}
\noindent As in the one-dimensional case, we state a stronger version of the previous refining theorem:
\begin{thm}
For any two potentials $f_1, f_2\in\mathcal{P}_d$ we have:
\begin{eqnarray}\label{new2}
\sigma(r) = \displaystyle\int_0^r[f_2(t) - f_1(t)]t^{d-1}R_j(t)dt\geq 0\quad r\in[0,\infty)\implies G_1(E)\leq G_2(E),
\end{eqnarray}
for $j = 1, 2$, where $E$ is the ground state energy $E$.
\end{thm}
\begin{proof}
In the same manner of the proof of the one-dimensional theorem we arrive to the following formula
\begin{eqnarray}\label{vad}
v_a = \dfrac{v\langle R_1, (f_2 - f_1)(vf - E)R\rangle}{\langle R_1, f(E - vf) R\rangle},
\end{eqnarray}
which is equal to
\begin{eqnarray}\label{vad1}
\dfrac{- \displaystyle\int_0^{\infty}\bigg[vf^{\prime}(r)R(r) + \big(vf(r) - E\big)R^{\prime}(r)\bigg]\sigma(r)dr}{\displaystyle\int_0^{\infty}R_1(r)\bigg[\sqrt{m^2 -\dfrac{R^{\prime\prime}(r)}{R(r)}-\dfrac{R^{\prime}(r)}{rR(r)}}\bigg]R(r)r^{d-1}f(r)dr}\geq 0.
\end{eqnarray}
Hence we have reached our desired result.
\end{proof}
% -----------------------------------------------------------------------------------------------------%
\subsection{Sign of Coulomb-like energy eigenvalues in dimension $d\geq 3$}
% -----------------------------------------------------------------------------------------------------%
In this section we study the sign of the energy eigenvalues of a certain class of Coulomb-like potentials in dimension $d\geq 3$.
% We first apply
% the change of variable $R(r) = r^{-\frac{d-1}{2}}\varphi(r)$ to (\ref{KGD}), to obtain the following reduced second-order differential equation:
%\begin{eqnarray}\label{KG4}
%\varphi^{\prime\prime}(r)=\bigg[ m^2-\big(E-vf(r)\big)^2+\frac{Q}{r^2}\bigg]  \varphi(r),
%\end{eqnarray}
%where 
%\begin{eqnarray*}
%Q=\frac{1}{4}(2l+d-1)(2l+d-3),
%\end{eqnarray*}
%with $l = 0, 1, 2, ...$. The reduced wave function $\varphi$ satisfies $\varphi(0) = 0$ and
%$\displaystyle\lim_{r\to\infty}\varphi = 0$ \cite{MMN}. For bound states, the normalization condition is:
%\begin{eqnarray*}
%\int_{0}^{\infty}\varphi^2 (r) dr= 1. 
%\end{eqnarray*}
\begin{thm}
Let $f\in\mathcal{P}_d$ such that $f(r) = -\dfrac{w(r)}{r}$ with $w(r)$ non-increasing, $w(0)\leq1$, and $\displaystyle\lim_{r\to\infty}w(r) = 0$. Then the corresponding ground state energy $E$ of equation (\ref{KGPSI}) is positive for $v < \dfrac{d-2}{2}$, $d\geq 3$.
\end{thm}
\begin{proof}
Multiplying equation (\ref{KGPSI}) by $\varphi$ and integrating over $[0,\infty)$ we get
\begin{eqnarray}\label{m1}
-2Ev\left\langle f\right\rangle =\left\langle -\Delta\right\rangle + m^2 - E^2 - v^2\left\langle f^2\right\rangle.
\end{eqnarray}
Using the generalized Heisenberg uncertainty relation for dimension $d\geq 3$, \bigg($\langle-\Delta\rangle\geq\left\langle\dfrac{(d-2)^2}{4r^2}\right\rangle$\bigg) (\cite{SI, Hardy1, Hardy2, Hardy3}), we obtain
\begin{eqnarray*}
-2Ev\left\langle f\right\rangle \geq \left\langle\dfrac{(d-2)^2}{4r^2}\right\rangle + m^2 - E^2 - v^2\left\langle f^2\right\rangle.
\end{eqnarray*}
Since $m^2 - E^2 > 0$ for all $E\in(-m,m)$ and $v < \dfrac{d-2}{2}$ then
\begin{eqnarray}\label{m2}
-2Ev\left\langle f\right\rangle > \dfrac{(d-2)^2}{4}\left\langle\dfrac{1}{r^2} - f^2\right\rangle.
\end{eqnarray}
Replacing $f(r)$ by $\dfrac{w(r)}{r}$ and using $d\geq 3$ in (\ref{m2}) we conclude that
\begin{eqnarray*}
-2Ev\left\langle f\right\rangle > \dfrac{1}{4}\left\langle\dfrac{1 - w^2(r)}{r^2}\right\rangle\geq 0.
\end{eqnarray*}
Hence, $E >e_0\geq 0$, for some non-negative real number $e_0$.
\end{proof}

We note here that the earlier refined  comparison theorems for the Klein--Gordon equation \cite{HP-16} required the energy $E$ to be positive. Since this condition is automatically satisfied by the present class of Coulomb--like potentials, the comparison theorems proved in Ref.\cite{HP-16} using the formulation $E = F(v)$ apply to these singular problems without modification.
%\newpage
% --------------------------------------------------------------------------%
%\begin{center}
%{\bf An Example in $d = 3$}
%\end{center}
% --------------------------------------------------------------------------%
%We consider the potentials $V_1(r) = v_1f_1(r)$ and $V_2(r) = v_2f_2(r)$ with $f_1(r) = -\dfrac{2e^{-0.35r}}{r}$, and $f_2(r) = -\dfrac{0.2}{e^{0.2r}-1}$. The respective graphs cross over at $r_0 \approx 2.82567$ ({\bf fig.$7$})and $\displaystyle\int_0^{r_0} \big[f_2(r) - f_1(r)\big]rdr \approx 1.13739$. We fix $E = 0.99158$ and we find that $v_1\approx 0.1863 < v_2\approx 0.22$. The graphs of $v_1(E)$ and $v_2(E)$ are shown in figure $8$.
%\begin{figure}
%\centering
%\includegraphics[scale=0.3]{pot-sing.pdf}
%\caption{Graphs of the potential shapes $f_1(r)=-\dfrac{2e^{-0.35r}}{r}$ and $f_2(x) = -\dfrac{0.2}{e^{0.2r}-1}$.}
%\end{figure}

%\begin{figure}
%\centering
%\includegraphics[scale=0.3]{spec-sing.pdf}
%\caption{Lower and upper graphs relative to  $f_1(r) = -\dfrac{2e^{-0.35r}}{r}$ and $f_2(r) = -\dfrac{0.2}{e^{0.2r}-1}$ respectively.}
%\end{figure}
%\newpage
% --------------------------------------------------------------------------%
\section{ Spectral bounds for given potential shapes}
% --------------------------------------------------------------------------%
% --------------------------------------------------------------------------%
\subsection{General spectral formulas provided by the square-well and the exponential potential.}% --------------------------------------------------------------------------%
%\newpage
In this section we provide a method for finding convenient lower and upper spectral bounds for a bounded potential shape $f(x)$. In the previous work we constructed both lower and upper bounds by fitting suitable square wells.  For the present application we have fitted a square well. We use the square-well potential and the exponential potential as a lower bound and an upper bound respectively. We have discussed a similar idea in our previous paper \cite{HH} based on square-well spectral bounds, but the results were limited by the condition that the graphs cannot cross over. Since we have been able to refine our previous comparison theorem, we can find better bounds now. We have chosen the square-well and the exponential potentials because we know the exact solutions of the Klein--Gordon equation with each of these potentials \cite{Gr, exp}, \cite{HH}. These solutions allow us to find the exact values for the corresponding eigenvalues $v=v(E)$.\\

\noindent Consider an attractive potential $V\in\mathcal{P}$ such that $V(r) = vf(r)$. Let $V_1(r) = v_1f_1(r)$ be the square-well potential such that 
\[ f_1(x) = 
\begin{cases}
           f(0), &  |x| \leq t\\
           0, & {\rm elsewhere}
      \end{cases}
\] 
with
\begin{eqnarray*}
 \int_0^t\big(f(r) - f_1(r)\big)dr> 0,
\end{eqnarray*}
 and  
\begin{eqnarray*}
\int_0^{\infty}\big(f(r) - f_1(r)\big)dr = 0. 
\end{eqnarray*}
We also consider the exponential potential $V_2(r) = v_2f_2(r)$ with $f_2(r) = -e^{-qr}$, $q > 0$, which intersects with $f$ at $r = \alpha$ such that 
\begin{eqnarray*}
\int_0^{\alpha}\big(f_2(r) - f(r)\big)dr > 0,
\end{eqnarray*}
 and
 \begin{eqnarray*}
\int_{0}^{\infty}\big(f_2(r) - f(r)\big)dr = 0.
\end{eqnarray*}
 Hence, for any eigenenergy $E\in(-m,\, m)$ we have $G_L(E)\leq G(E)\leq G_U(E)$ where $G_L, G$, and $G_U$ are the respective graphs of the spectral functions $v_1(E), v(E)$, and $v_U(E)$ respectively.
\newpage
% --------------------------------------------------------------------------%
\begin{center}
{\bf Applications}
\end{center}
% --------------------------------------------------------------------------%
\begin{enumerate}
\item Let $V(x) = v\,f(x)$ be the Gaussian potential where $f(x) = -e^{-qx^2}$, and $q > 0$ is a range parameter. We want to find a lower and an upper bound for the coupling constant $v$, for any given eigenenergy $E\in(-m, m)$ and $q = -0.8$. We choose the square-well potential $V_1(x) = v_1f_1(x) $ and the exponential potential $V_2(x) = v_2f_2(x)$ with 
\begin{equation*}
\centering
{f_1(x) = \begin{cases}
                    -1, &  |x| \leq \frac{\sqrt{5\pi}}{4}\\
                     0, & {\rm elsewhere}
              \end{cases}} ~~~~{\rm and~~}  f_2(x) =  -e^{-\frac{4}{\sqrt{5 \pi}}|x|}.   
\end{equation*}   
We have $\int_0^{\frac{\sqrt{5\pi}}{4}}\big(f(x) - f_1(x)\big)dx \approx 0.20816$ and 
$\int_0^{\infty}\big(f(x) - f_1(x)\big)dx=0$~(Figure~$1$). On the other hand, $f$ and $f_2$ cross over at $x_0 \approx 1.26$ ( Figure~$2$) with $\int_0^{x_0}\big(f_2(x) - f(x)\big)dx\approx 0.15253$ and $\int_0^{\infty}\big(f_2(x) - f(x)\big)dx = 0$. We fix $E = -0.0377$ and we deduce that $v_1\leq v\leq v_2$ where $v_1 = 1.36$ and $v_2 = 1.9$. We have verified this result numerically by using our own shooting method realized in Maple, and with which we find $v = 1.581$. The graphs of $v_1(E), v(E)$, and $v_2(E)$ are shown in Figure~$3$. References for shooting methods are found in \cite{JNKutz, JDPyrce}.
%\newpage
\begin{figure}
\centering
\includegraphics[scale=0.35]{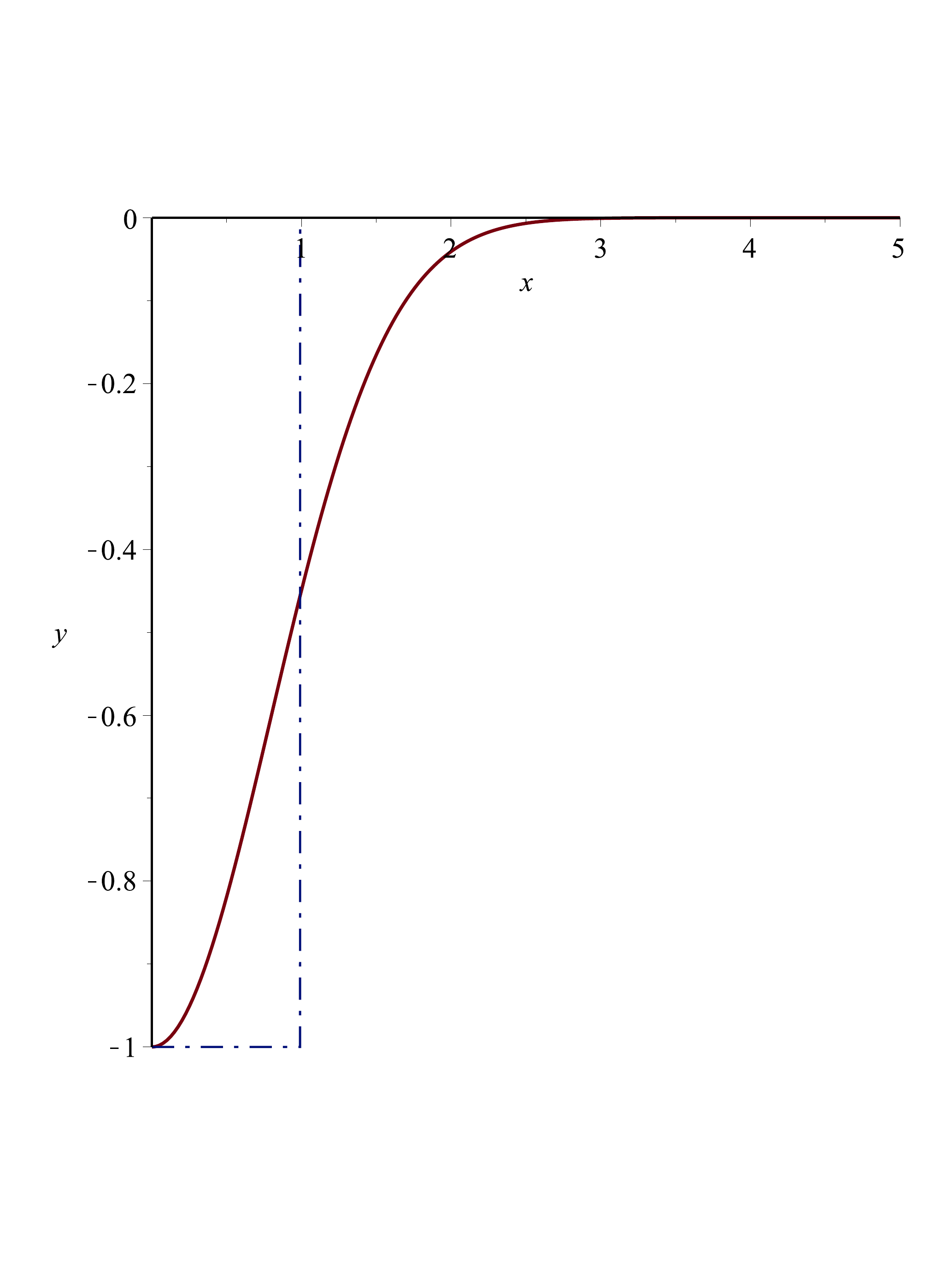}
\caption{Potential Shapes $f_1(x) = -1$ if $|x|\leq\frac{\sqrt{5\pi}}{4}$ and $0$ elsewhere, 
dashed lines and $f(x) = -Ae^{-qx^2}$ full line, where $A =1$ and $q = 0.8$ were applied.}
\end{figure}
%\newpage
\begin{figure}
\centering
\includegraphics[scale=0.35]{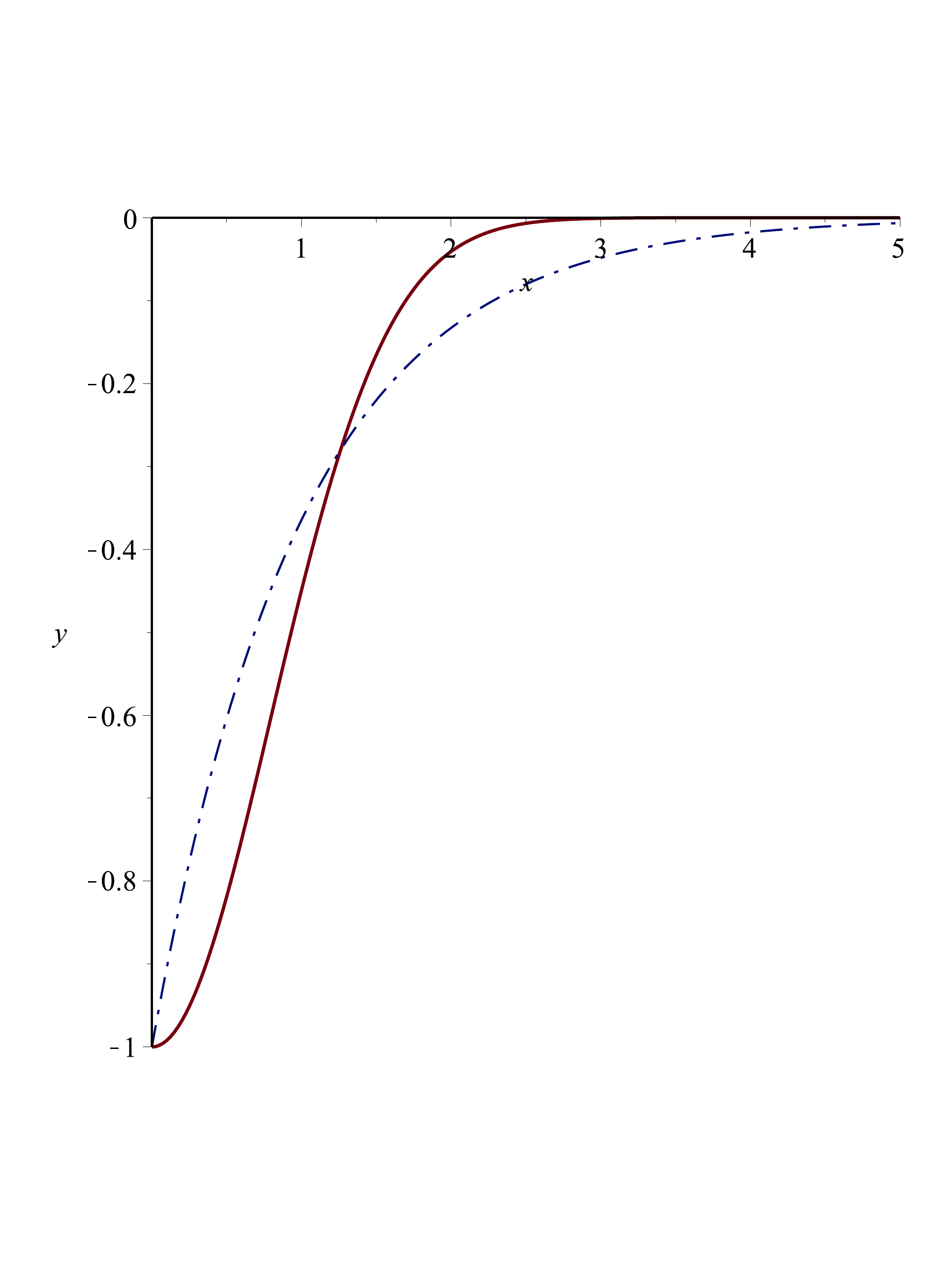}
\caption{Potential Shapes $f(x) = -Ae^{-qx^2}$ dashed lines and $f_2(x) = -Be^{-a|x|}$ full line, where $q = 0.8$, $a = \frac{4}{\sqrt{5\pi}}$, and $A = B = 1$ were applied.}
\end{figure}
%\newpage
\begin{figure}
\centering
\includegraphics[scale=0.35]{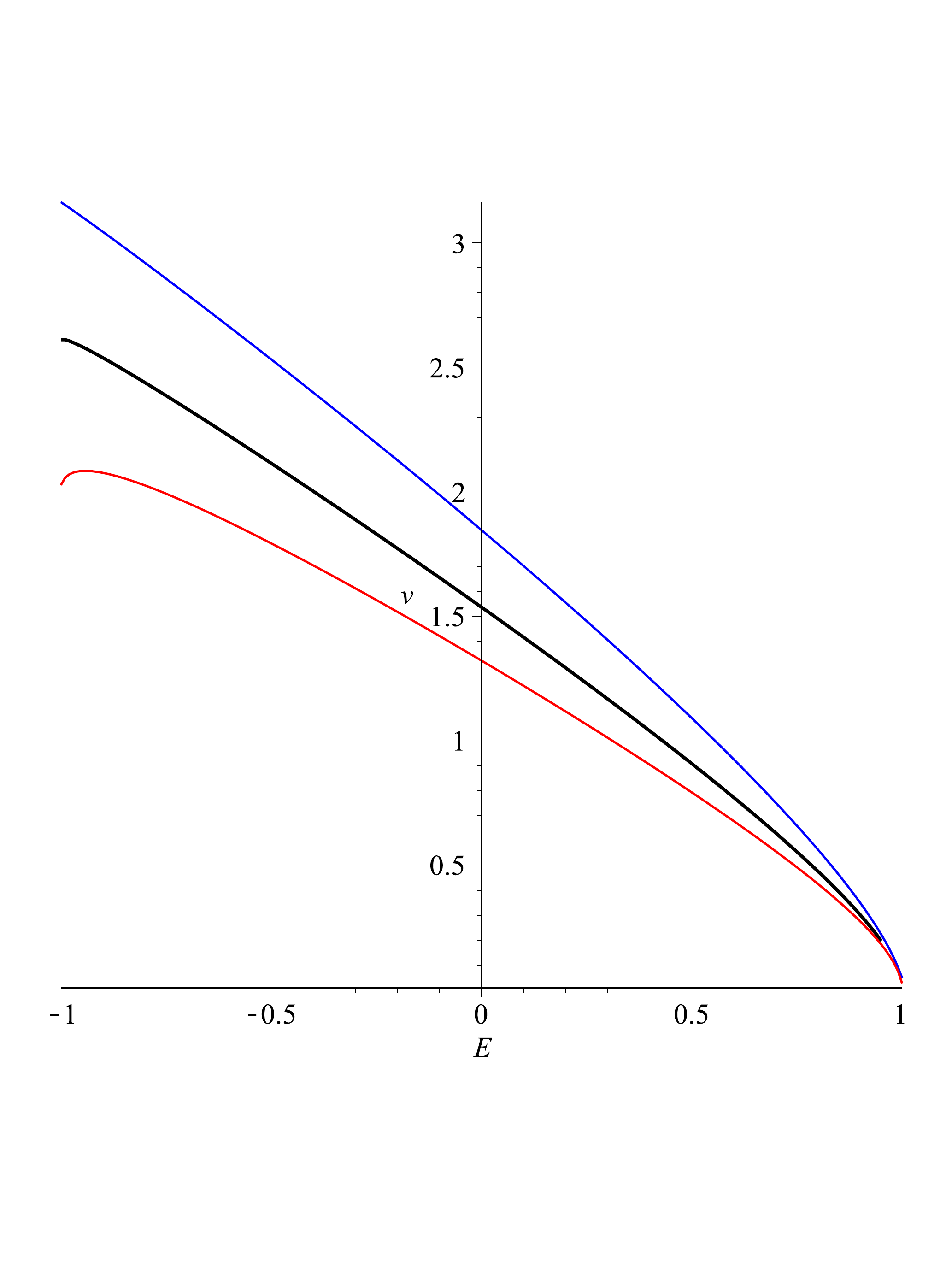}
\caption{Graphs for $v_1(E), v(E)$, and $v_2(E)$ corresponding to $f_1(x) = -1$ if $|x|\leq\frac{\sqrt{5\pi}}{4}$ and $0$ elsewhere, $f(x) = -e^{-0.8x^2}$, and $f_2(x) = -e^{-\frac{4|x|}{\sqrt{5\pi}}}$ respectively, for $-1 < E < 1$.}
\end{figure}
%\newpage
\item In this example we consider the sech-squared potential $V(x) = v\,f(x)$ where $f(x) = -\dfrac{\beta}{(e^{-q|x|}+e^{q|x|})^2}$\cite{PSE, CE, ETG, BYT}. We now find the spectral bounds for $\beta = 3$. We choose the exponential potentials $V_1(x) = v_1f_1(x)$ and $V_2(x) = v_2f_2(x)$ with the following parameters 
\begin{equation*}
f_1(x) = -e^{-0.46666|x|}~~{\rm and~~} f_2(x) = -0.75e^{-0.35|x|}.
\end{equation*}
This example illustrates that application the refinement theorem when the corresponding potential shapes cross over more than once:  as long as the integral of their difference is convergent, we obtain spectral bounds.  Figure~$4$ and Figure~$5$ show how the graphs of  the pairs $\{f_1, f\}$, and $\{f,f_2\}$ cross over: 
\begin{eqnarray*}
\displaystyle\int_0^{\infty}\big(f(x) - f_1(x)\big)dx = \displaystyle\int_0^{\infty}\big(f_2(x) - f(x)\big)dx = 0.
\end{eqnarray*}
We fix $E = -0.314$ and we get $v_1 = 1.9\leq v\leq v_2 = 2.39$. We have verified this result numerically, by using our own shooting method, and we find find that $v = 2.0943$.
The graphs of $v_1, v$, and $v_2$ are shown in Figure~$6$: this data is obtained from the exact analytical solutions, mentioned above.
%\newpage
\begin{figure}
\centering
\includegraphics[scale=0.35]{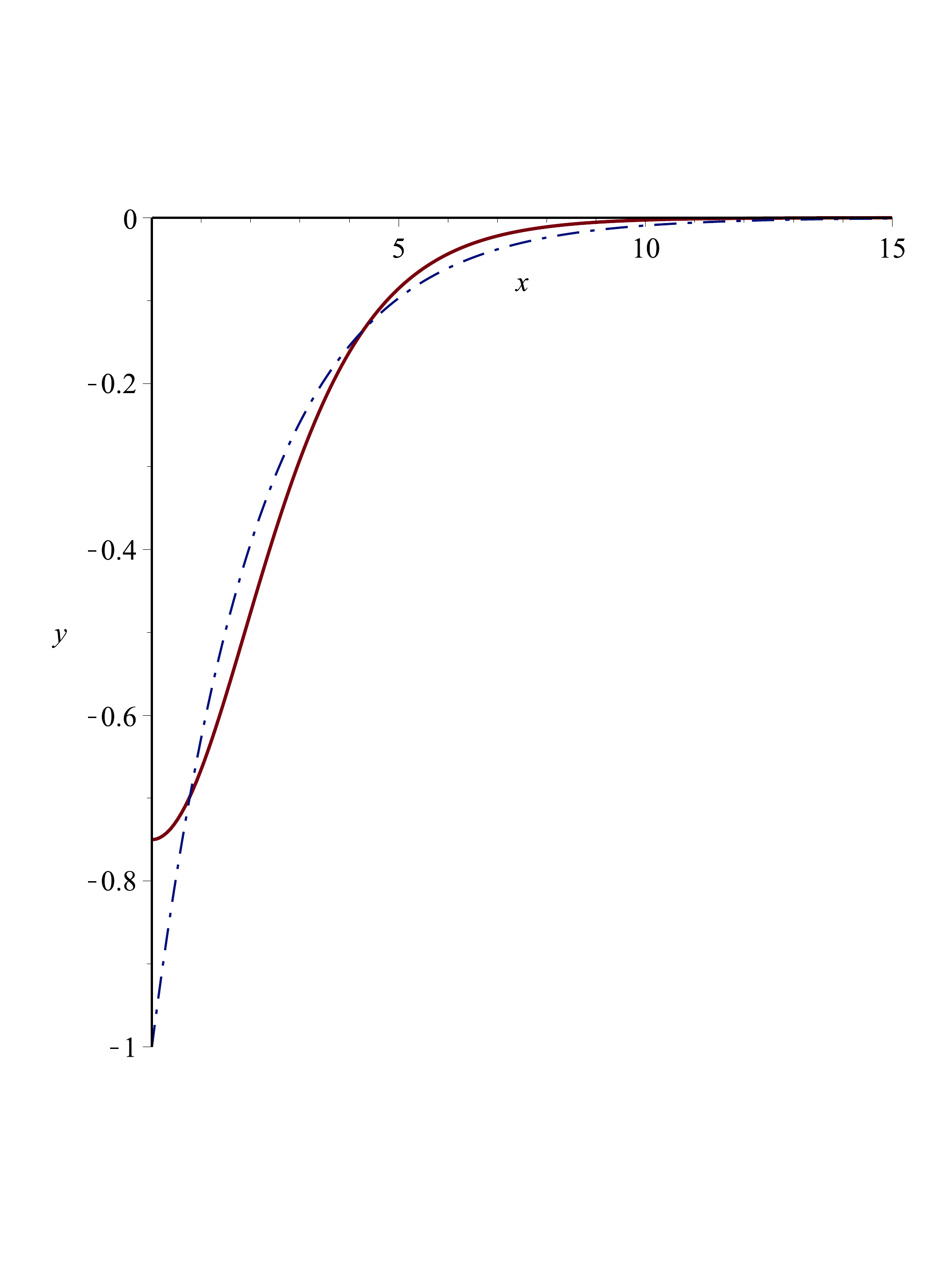}
\caption{Potential Shapes $f_1(x) = -Ae^{-q|x|}$ dashed lines and $f(x) = -\dfrac{\beta}{(e^{-q|x|}+e^{q|x|})^2}$ full line, where $q = 0.35$, and $\beta = A = 1$ were applied.}
\end{figure}
%\newpage
\begin{figure}
\centering
\includegraphics[scale=0.35]{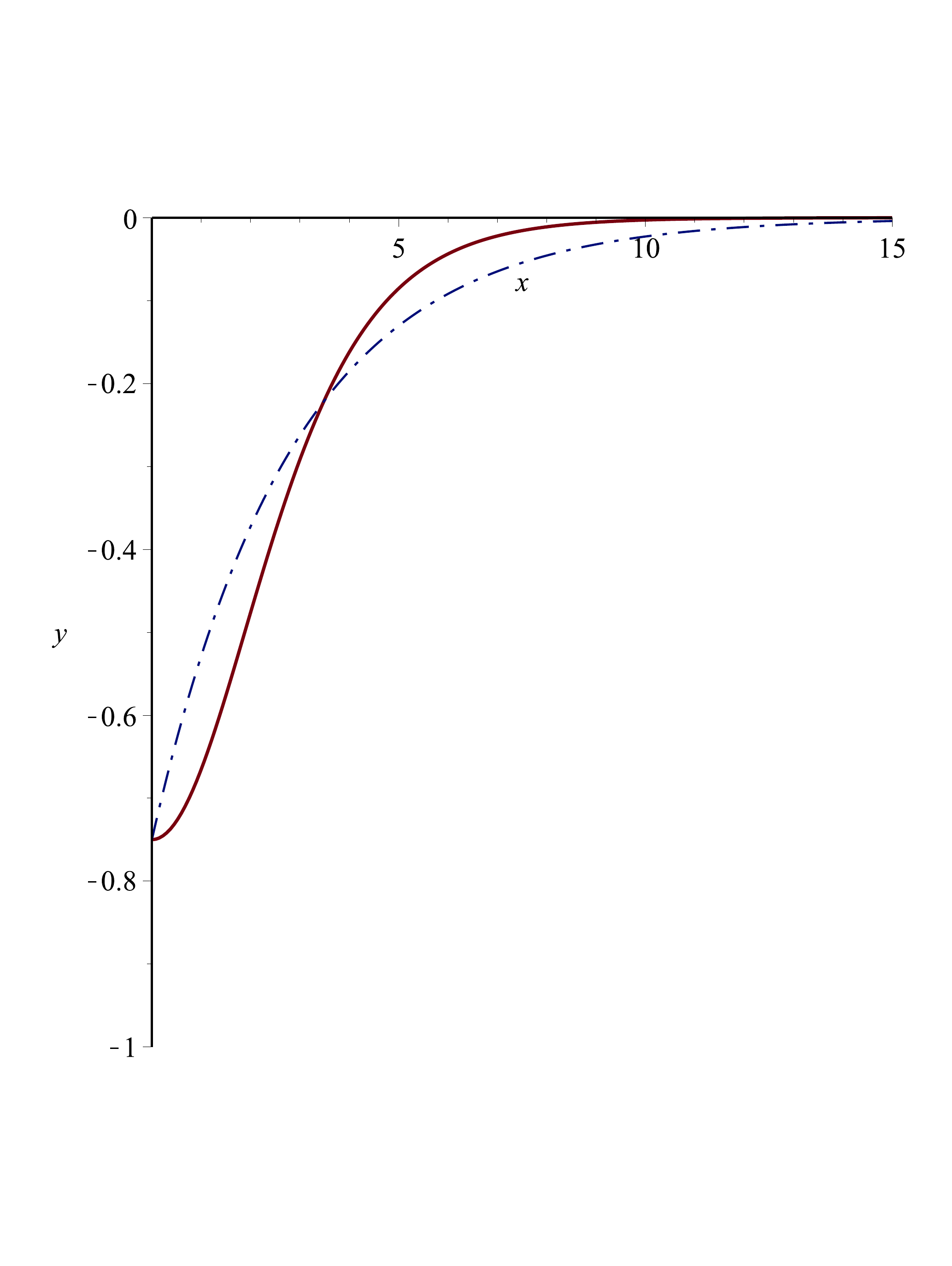}
\caption{Potential Shapes $f_2(x) = -Ae^{-q|x|}$ dashed lines and $f(x) = -\-\dfrac{\beta}{(e^{-0.35|x|}+e^{0.35|x|})^2}$ full line, where $q = 0.35$, $\beta = 1$, and $A = 0.75$ were applied.}
\end{figure}
%\newpage
\begin{figure}
\centering
\includegraphics[scale=0.35]{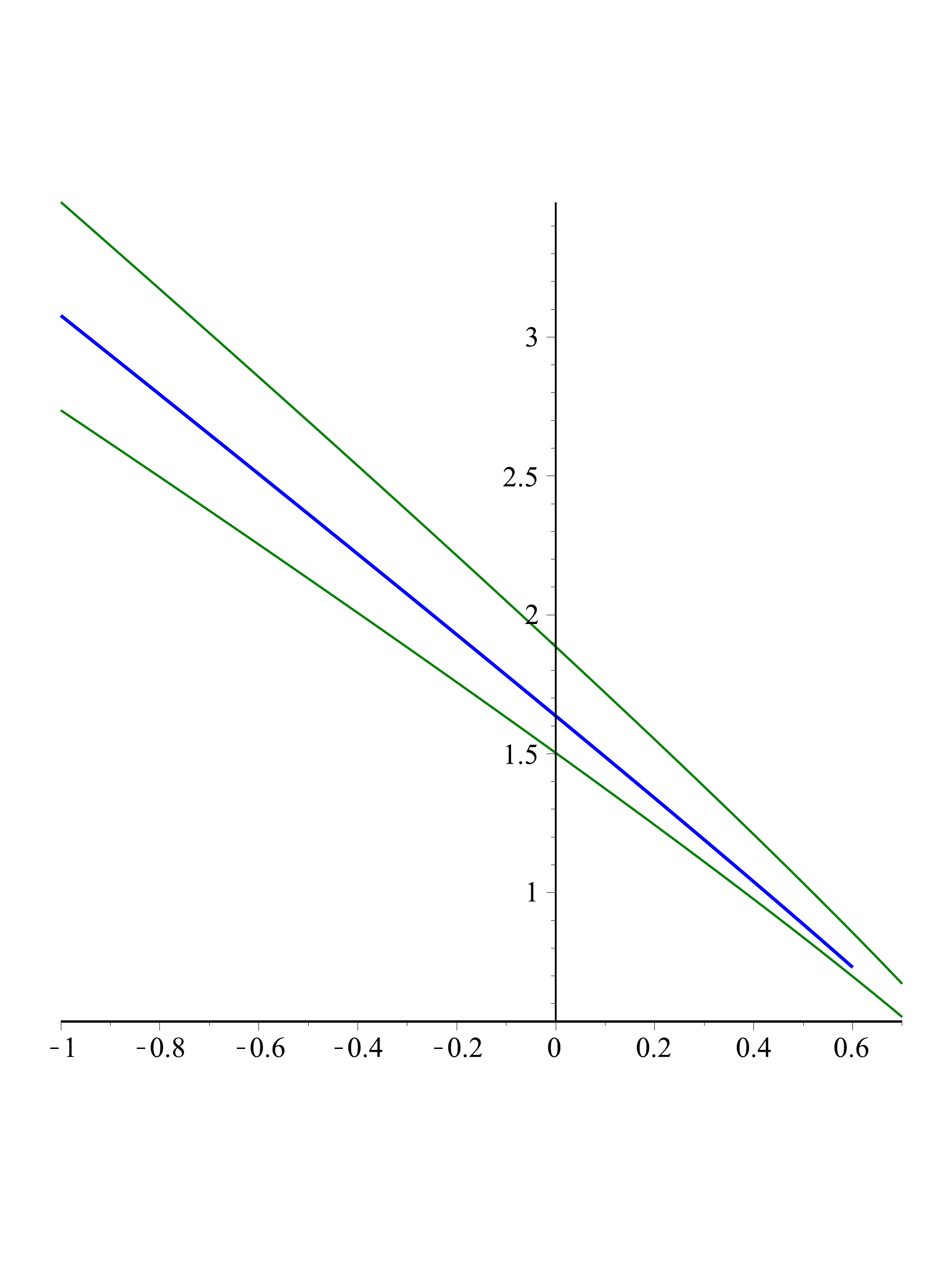}
\caption{Graphs for $v_1(E), v(E)$, and $v_2(E)$ corresponding to $f_1(x) = -e^{-0.46666|x|}$, $f(x) = -\dfrac{\beta}{(e^{-0.35|x|}+e^{0.35|x|})^2}$, and $f_2(x) = -0.75e^{-0.35|x|}$ respectively, for $-1 < E < 1$.}
\end{figure}

\end{enumerate}
%\newpage
% ---------------------------------------------------------------------------------------------------%
\subsection{Hulth\'en and Coulomb spectral bounds for singular potentials}
% ---------------------------------------------------------------------------------------------------%
Since we know the exact solutions of the Klein--Gordon equation with the Coulomb and Hulth\'en potentials \cite{Gr1, SQ, LIS, FDA}, we can find spectral bounds for any singular potential in $\mathcal{P}_d$.
% --------------------------------------------------------------------------%
\begin{center}
{\bf The Yukawa Potential in dimension $d = 3$}
\end{center}
% --------------------------------------------------------------------------%
Consider the Yukawa potential $V(r) = vf(r)$ with $f(r) = -\dfrac{e^{-ar}}{r}$ \cite{Yuk} , where $a>0$ is a range parameter. We shall find a lower and an upper bound for the coupling constant $v$, for any $E\in(-m,\,m)$, for $a = 0.5$. We choose the Hulth\'en potentials $V_1(r) = v_1f_1(r)$ and $V_2(r) = v_2f_2(r)$ where \cite{LH1, LH2, MSLH}
\begin{equation*}
f_1(r) = -\frac{1}{e^{1.001r}-1}~{\rm  and}~f_2(r) = -\frac{1}{e^{0.966r}-1}.
\end{equation*}
We fix $E = 0.96$ and obtain $v_1 = 0.4895$ and $v_2 = 0.4799$.
Since $f_1(r) > f(r)$ for $r\in[0,\infty)$ as shown in Figure~$7$, then according to our simple general comparison theorem \cite{HH}, we find that $v_1 > v$. On the other hand, $f$ and $f_2$ cross over at $r_0\approx 1.2$ as shown in the right graph of Figure~$8$,  with $\int_0^{r_0}\big(f_2(r) - f(r)\big)r^2dx = 0.0108 > 0$. Hence, applying our refined version of the general comparison theorem, we obtain $v> v_2$. We have numerically verified this result by finding that $v = 0.4834$. The graphs of $v_1(E), v(E)$, and $v_2(E)$ are shown in Figure~$9$.
%\newpage

\begin{figure}
\centering
\includegraphics[scale=0.35]{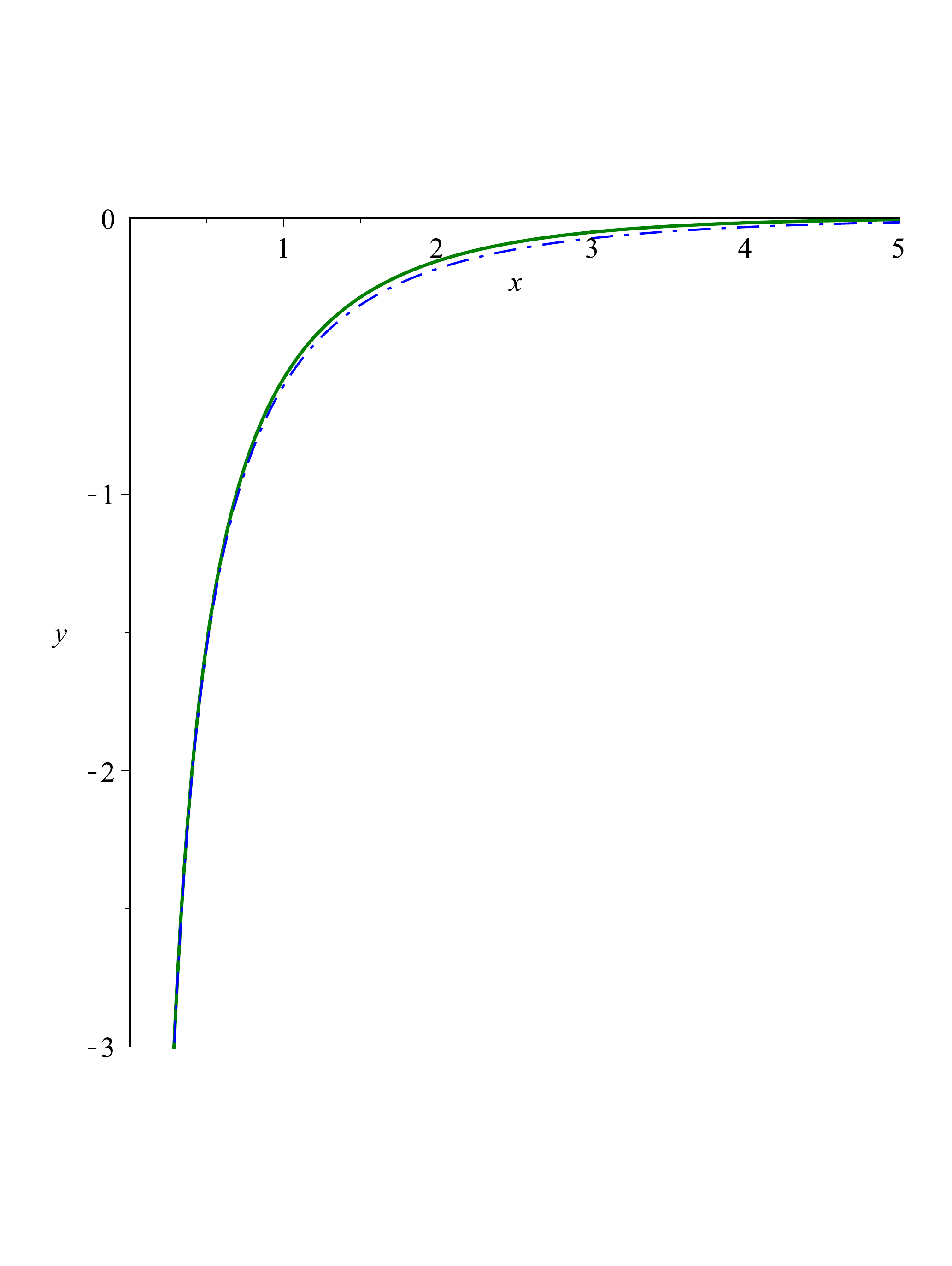}
\caption{Potential Shapes $f_1(r) = -\dfrac{1}{e^{1.001r}-1}$ full line and $f(r) = -\dfrac{e^{-ar}}{r}$ dashed lines, where $a = 0.5$ was applied.}
\end{figure}

\begin{figure}[!tbp]
\centering
\begin{minipage}[b]{0.4\textwidth}
\includegraphics[width=\textwidth]{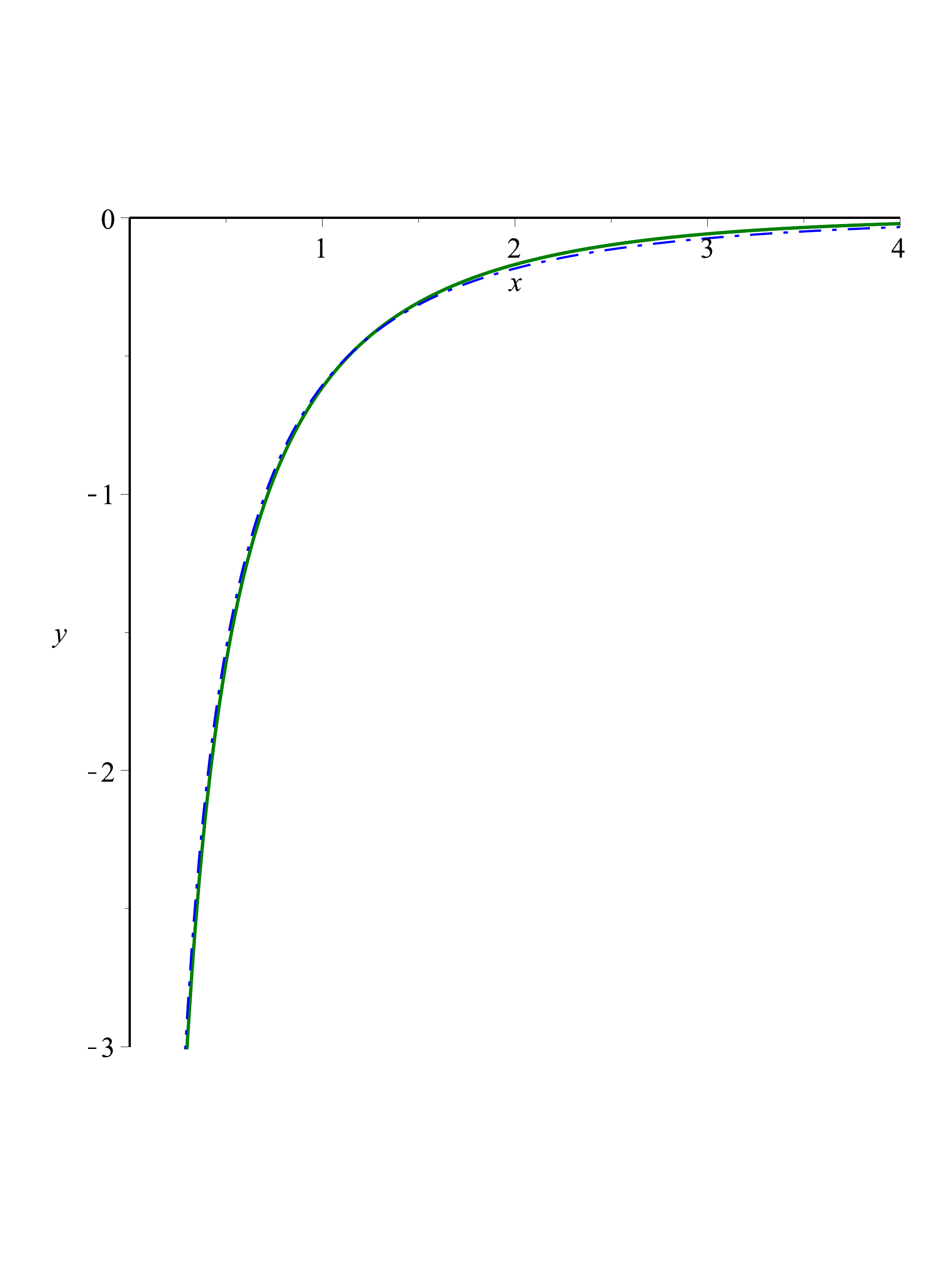}
%\caption{Potential Shapes $f_1(x) = -\dfrac{1}{e^{0.3r}-1}$, dashed lines and $f(x) = -\dfrac{e^{-ax^2}}{r}$ full line, where $a = 0.014$ was applied.}
\end{minipage}
\hfill
\begin{minipage}[b]{0.4\textwidth}
\includegraphics[width=\textwidth]{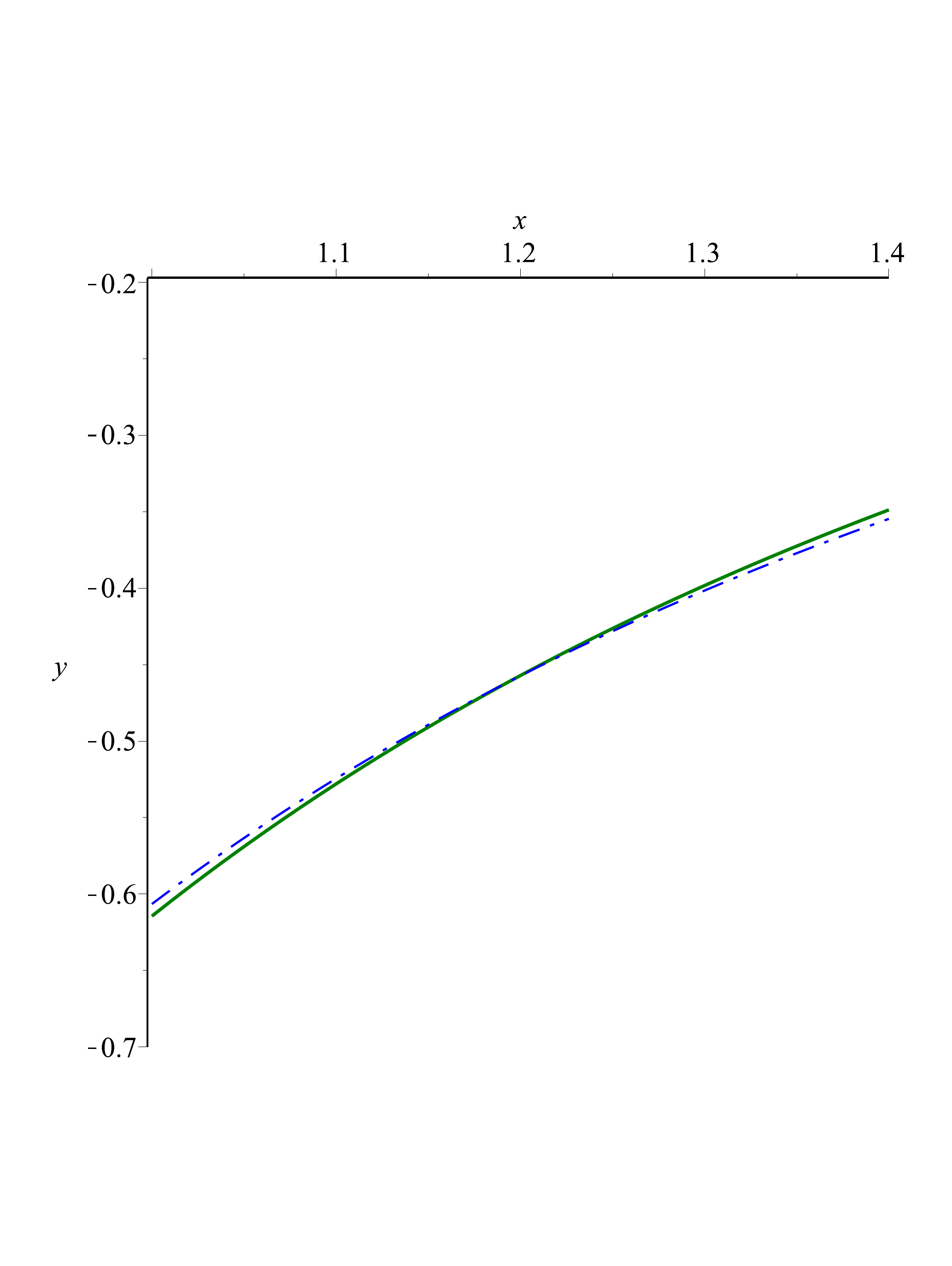}
%\caption{Potential Shapes $f_1(x) = -\dfrac{1}{e^{0.3r}-1}$, dashed lines and $f(x) = -\dfrac{e^{-ax^2}}{r}$ full line, where $a = 0.014$ was applied.}
\end{minipage}
\caption{Left graph: potential shape $f_2(r)=-\dfrac{1}{e^{0.966}-1}$ (solid line)  and $f(r) =-\dfrac{e^{-ar}}{r}$ (dashed lines) . They intersect at $r_0\approx 1.2$ as shown in the right graph. $a = 0.5$ was applied.}
\end{figure}

\begin{figure}
\centering
\includegraphics[scale=0.35]{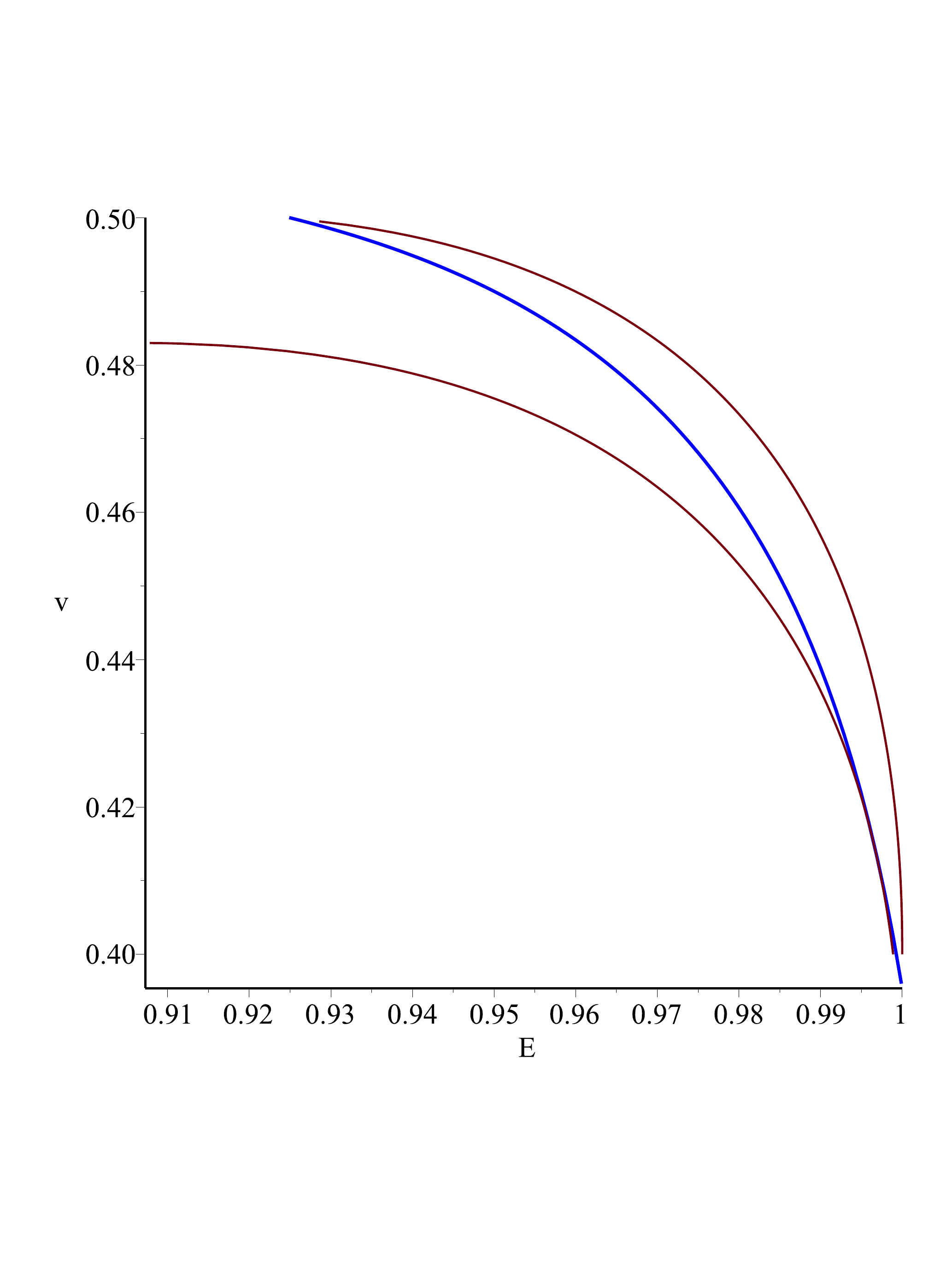}
\caption{Graphs for $v_1(E), v(E)$, and $v_2(E)$ corresponding to $f_1(r) = -\dfrac{1}{e^{1.001r}-1}$, $f(r) = -\dfrac{e^{-0.5r}}{r}$, and $f_2(r) = -\dfrac{1}{e^{0.966r}-1}$ respectively, for $-1 < E < 1$.}
\end{figure}
\newpage
%\begin{figure}
%\centering
%\includegraphics[scale=0.33]{spec-sq-gauss-exp.pdf}
%\caption{Graphs for $v_1, v$, and $v_2$ versus $E$ for $-1 < E < 1$.}
%\end{figure}
\section{conclusion}
We have shown in this paper that our general comparison theorem for the Klein--Gordon equation~\cite{Aliyu},  $f_1\leq f_2\implies v_1\leq v_2$, still holds for the nodeless states, even if the condition is weakened to $\int_0^x\big[f_2(t) - f_1(t)\big]dt\geq 0$ for $d = 1$, and 
$\int_0^r\big[f_2(t) - f_1(t)\big]t^{d-1}dt\geq 0$ for $d\geq 3$, on $[0,\infty)$. Thus, the comparison potentials are allowed to cross over in a controlled way, even 
many times. We have called this kind of result a `refined comparison theorem'.  The comparison theorems in \cite{Aliyu} and \cite{HH} had a limitation, because the proof would only go through if $E$ was negative. These difficulties have been completely overcome in the present paper. The proofs are now valid for $E$ positive or negative.     We have also proven that if we know one of the wave functions $\varphi_1$ or $\varphi_2$, we can replace the above condition by $\int_0^x\big[f_2(t) - f_1(t)\big]\varphi_i(t)dt\geq 0$ for $d = 1$, and $\int_0^r\big[f_2(t) - f_1(t)\big]t^{d-1}\varphi_i(t)dt\geq 0$ for $d\geq 3$, with $i = 1,2$. The latter conditions provide us with a stronger theorem because, since the ground state is non-increasing on $[0,\infty)$, the potential shapes are allowed to cross over `even more' while preserving the ordering of the coupling parameters $v_1 = G_1(E)$ and $v_2=  G_2(E)$, for any energy $E\in(-m,\,m)$.  We have also proven that for a potential whose shape $f(r)$ is no more singular than $r^{-(d-2)}$, $(d\geq 3)$, with $f(r) = -\frac{w(r)}{r}$, where $w(r)$ is non-increasing on $[0,\infty)$ and $0 < w(0)\leq 1$, the lowest  eigenvalue                                                                                                                                                                                                                                                                      is always positive for $v < \frac{1}{2}$. As an application of our refined theorem, we have constructed upper and lower bounds for the discrete spectrum generated by a given central negative bounded potential, by using the known exact solutions of the Klein--Gordon equation with square-well and exponential potentials. 

\vskip0.2cm
Comparison theorems can be thought of as a contribution to spectral approximation theory.  This would be a rather narrow view.  A more complete description would perhaps use terms like `functional perturbation theory' which would be suggestive of what one actually gets, namely a set of almost instant approximations  
allowing for adaptive approaches to model building, as we already have in non-relativistic quantum mechanics.

\newpage
% ------------------------------------------------------
 \begin{acknowledgments}
% ------------------------------------------------------
Partial financial support of his research under Grant No.\,GP3438 from the Natural Sciences and Engineering Research Council of Canada is gratefully acknowledged by one of us (RLH).
 \end{acknowledgments}

\end{document}